\documentclass{article}
\usepackage{amsmath,amsthm,amssymb}
\usepackage{array, longtable}
\newtheorem{Theorem}{Theorem}[section]
\newtheorem{lem}[Theorem]{Lemma}
\newtheorem{Remark}[Theorem]{Remark}

\newtheorem{Corollary}[Theorem]{Corollary}
\newtheorem{Proposition}[Theorem]{Proposition}
\newtheorem{Example}[Theorem]{Example}

\numberwithin{equation}{section}
\numberwithin{table}{section}

\begin{document}

\title{Constructions of  cyclic constant dimension codes }

\insert\footins{\footnotesize {\it Email addresses}:
{bocong\_chen@yahoo.com (B. Chen)},
hwliu@mail.ccnu.edu.cn (H. Liu). }

\author{Bocong Chen$^{1}$,~~Hongwei Liu$^{2}$}

\date{\small $^{1}$  School of Mathematics, South China University of Technology, Guangzhou,
Guangdong, 510641, China\\
${}^2$School of Mathematics and Statistics,
Central China Normal University,
Wuhan, Hubei, 430079, China}

\maketitle

\begin{abstract}
Subspace codes and particularly constant dimension codes have attracted much attention in recent years
due to their applications in random network coding. As a particular subclass of subspace codes,
cyclic subspace codes have   additional properties
that can be applied efficiently  in  encoding and decoding algorithms.
It is
desirable to find   cyclic constant dimension codes
 such that both  the code sizes and the minimum   distances
are as large as possible. In this paper, we explore the ideas  of constructing cyclic constant dimension codes proposed in
\big(\cite{Ben},  IEEE Trans. Inf. Theory,   2016\big) and  \big(\cite{OO16}, Des. Codes Cryptogr.,
 2016\big) to obtain further results.
Consequently, new code constructions  are provided  and  several  previously known results in \cite{Ben} and \cite{OO16} are extended.


\end{abstract}

\section{Introduction}
Let $\mathbb{F}_q$ be the finite field of size $q$   and let $\mathbb{F}_{q^{N}}$ be the field extension of degree $N$ over $\mathbb{F}_{q}$;
$\mathbb{F}_{q^{N}}$ can be viewed as an $N$-dimensional vector space over $\mathbb{F}_q$.
The set of all subspaces of $\mathbb{F}_{q^{N}}$,
denoted  by $\mathcal{P}_q(N)$,  is called the {\em projective space} of order $N$ over $\mathbb{F}_q$ (see \cite{Etzion2}).
For any $U,V\in \mathcal{P}_q(N)$, the {\em subspace distance}  between $U$ and $V$ is defined to be
\begin{equation*}
\begin{split}
d(U,V)=\dim\big(U+V\big)-\dim\big(U\bigcap V\big)
=\dim\big(U\big)+\dim\big(V\big)-2\dim\big(U\bigcap V\big).
\end{split}
\end{equation*}
It turns out that the set $\mathcal{P}_q(N)$ equipped with the subspace distance is indeed a {\em metric space} (see \cite{Koetter}).
A {\em subspace code} $\mathcal{C}$ is simply a nonempty subset  of $\mathcal{P}_q(N)$;
if,  in addition,  all the elements of $\mathcal{C}$ have the same dimension $k$, then
$\mathcal{C}$ is called a  {\em $k$-dimensional subspace code} (or {\em constant dimension code} for short).
The {\em minimum (subspace) distance} of any subspace code $\mathcal{C}$ is defined to be
$$
d(\mathcal{C})=\min\limits_{U\neq V\in \mathcal{C}} d(U,V).
$$

Subspace codes and particularly constant dimension codes have attracted much attention in recent years
due to their applications in random network coding for correction of  errors and erasures \cite{Koetter}.
Subspace codes are  also of interest from mathematical
viewpoints  (see, for example, \cite{Ahl}, \cite{Chi}, \cite{Martin},  \cite{Sch}  and \cite{Xia}).
One of the main research problems on constant dimension codes is to  find systematic methods to produce good $k$-dimensional subspace codes with
a large code size and a large minimum  distance
when $q, N, $ and $k$ are fixed.
The   seminal works \cite{Koetter,Silva08}
presented   novel constructions of large constant dimension codes through
linearized polynomials.

As a particular subclass of subspace codes, cyclic subspace codes have   additional properties
that can be applied efficiently  in  encoding and decoding algorithms (e.g., see \cite{Tra}).
For a given subspace $U\in\mathcal{P}_q(N)$ and  $\alpha\in\mathbb{F}_{q^N}^*$  (where $\mathbb{F}_{q^N}^*=\mathbb{F}_{q^N}\setminus\{0\}$),
the {\em cyclic shift}
of $U$ is defined by $\alpha U=\{\alpha u\,|\,u\in U\}$, where the product $\alpha u$ is taken in $\mathbb{F}_{q^N}$.
It is clear that $\alpha U$ is a vector space over $\mathbb{F}_q$ having the same dimension as $U$.
Two cyclic shifts are called {\em distinct}
if they form two different subspaces.
A subspace code $\mathcal{C}$ is said to be {\em cyclic} if
$\alpha U\in \mathcal{C}$ for any   $\alpha\in\mathbb{F}_{q^N}^*$ and any $U\in \mathcal{C}$.
Several  optimal cyclic subspace
codes with small dimensions were found in \cite{Etzion2} and \cite{Koh}.
A thorough analysis of the algebraic structure of
cyclic subspace codes was given in \cite{Tra}. In \cite{BEOV13}, an optimal code which also forms a $q$-analog of Steiner system was presented.

We  tacitly assume $k>1$, as the case $k=1$ is uninteresting.
The biggest possible value for the minimum   distance of any $k$-dimensional  cyclic subspace code  is $2k$. However,
it is not hard to see that if the size of a $k$-dimensional cyclic  subspace code is greater than or equal to $(q^N-1)/(q-1)$ then its minimum  distance cannot achieve $2k$. In other words,  the best minimum  distance of a $k$-dimensional cyclic
subspace code whose size is greater than or equal to $(q^N-1)/(q-1)$  can attain is thus $2k-2$.
By virtue of this fact and using computer searches,
\cite{Tra} and \cite{Glu} raised the following conjecture: For any positive integers $N$ and $k$ with
$k<N/2$, there exists a $k$-dimensional cyclic subspace code of size $(q^N-1)/(q-1)$ and of minimum distance $2k-2$.
As mentioned at the end of \cite{Glu}, it would also be interesting to  find systematic
methods to produce $k$-dimensional cyclic subspace codes whose sizes exceed   $(q^N-1)/(q-1)$ and whose minimum distances
remain exactly $2k-2$.

Recently, Ben-Sasson {\em et al.} \cite{Ben}
used subspace polynomials to generate $k$-dimensional
cyclic subspace codes with size $(q^N-1)/(q-1)$ and minimum distance $2k-2$:
Let $V$ denote  the  set of roots of the trinomial
$X^{q^k}+X^{q}+X\in \mathbb{F}_{q}[X]$ (suppose $V$ is contained in $\mathbb{F}_{q^N}$). Then
$\{\alpha V\,|\,\alpha\in \mathbb{F}_{q^N}^*\}$
is a cyclic subspace code of which the size is $(q^N-1)/(q-1)$ and the minimum distance is $2k-2$. The conclusion of this result reveals
that
the aforementioned conjecture holds true for any  given $k$
and infinitely many values of $N$.
Furthermore, in the same paper the authors provided a construction of  $k$-dimensional cyclic subspace  codes of  size $r\frac{q^N-1}{q-1}$
and  minimum distance $2k-2$, which is the first systematic construction of cyclic constant dimension codes of size greater than
$(q^N-1)/(q-1)$.
Otal and \"{O}zbudak \cite{OO16} generalized and improved the construction in
\cite{Ben} by studying the roots of the trinomials
$X^{q^k}+\theta_i X^q+\gamma_i X\in \mathbb{F}_{q^n}[X]$,
where $\theta_i$ and $\gamma_i$ are
nonzero elements of  $\mathbb{F}_{q^n}$  for $1\leq i\leq r$.
As a consequence, some constraint conditions in \cite{Ben} are relaxed, the density of the length parameter $N$ is increased,
and the size of
$k$-dimensional cyclic subspace codes can be increased up to $(q^n-1)\frac{q^N-1}{q-1}$ without  decreasing the
minimum  distance $2k-2$.

The present paper is to extend the previous works  \cite{Ben} and \cite{OO16} by proposing a different approach.
We   explore the ideas  of constructing cyclic constant dimension codes proposed in
\cite{Ben} and \cite{OO16} to obtain further results;
consequently,  new code constructions  are provided and
several  previously known results in \cite{Ben} and \cite{OO16} are extended.
More explicitly, we show that if
the  set of roots of the trinomial
$X^{q^k}+a_\ell X^{q^{\ell}}+a_0 X\in \mathbb{F}_{q^n}[X]$
is denoted by $V$ (suppose $V$ is contained in $\mathbb{F}_{q^N}$),
where $1\leq \ell<k$ is a positive integer relatively prime to $k$, $n$ is an arbitrary positive integer and $a_0,a_\ell$ are nonzero elements
of $\mathbb{F}_{q^n}$,  then
$\{\alpha V\,|\,\alpha\in \mathbb{F}_{q^N}^*\}$
is a cyclic subspace code of which the size is $(q^N-1)/(q-1)$ and the minimum distance is $2k-2$ (see Lemma \ref{generalized1}).
Moreover, unions of such cyclic constant dimension codes from the roots of trinomials and binomials are also discussed
(see  Theorem \ref{theoremend} and its corollaries). Several examples are   provided to illustrate our results.
We mention that we can produce an infinite family of  $k$-dimensional cyclic subspace codes with size
$(q^n-1)\frac{(q^N-1)}{q-1}+\frac{q^N-1}{q^k-1}$ and minimum distance $2k-2$ (see Corollary \ref{qn}).

The remainder of this paper is organized as follows. Section \ref{sec-pr}
establishes some  notations that will be used throughout,  and reviews some basic  results that will be needed in subsequent sections,
including the notions of linearized polynomials and subspace polynomials.
Section \ref{sec-main} contains our main results. Section \ref{sec-concluding} summarizes this paper.

\section{Preliminaries}\label{sec-pr}
A class of polynomials that plays
an important role in the study of subspace  codes is the so-called  linearized
polynomials (e.g., see \cite[P. 107]{Lidl}). In this section we will briefly review
the definitions and some basic properties about linearized polynomials.

Throughout this paper,  $\mathbb{F}_q$ denotes the finite field of size $q$.
Let $n\geq1$ be a positive integer and let $\mathbb{F}_{q^n}$ be the field extension of degree $n$ over $\mathbb{F}_q$.
Recall that $\mathcal{P}_q(N)$ denotes the set of all subspaces of $\mathbb{F}_{q^N}$, where $N\geq1$ is an integer.
A {\em linearized polynomial} over $\mathbb{F}_{q^n}$ is a polynomial of the form
\begin{eqnarray*}
f(X)=\alpha_k X^{q^k}+\alpha_{k-1} X^{q^{k-1}}+\cdots + \alpha_{1} X^{q}+ \alpha_{0} X\in \mathbb{F}_{q^n}[X],
\end{eqnarray*}
where $\alpha_i$ are elements of $\mathbb{F}_{q^n}$ for $0\leq i\leq k$.
If $\alpha_k\neq0$ then  $k$ is called the {\it $q$-degree} of $f$.
Linearized polynomials have   the following
properties (see \cite{Lidl}):
\begin{Proposition}
The roots of any linearized polynomial form a subspace in some
extension field over $\mathbb{F}_{q^n}$.
Conversely, for any  subspace $V\in \mathcal{P}_q(N)$, the polynomial
$$
\prod_{v\in V}(X-v)
$$
is a linearized polynomial.
\end{Proposition}

It is well known that a linearized polynomial has no multiple roots if and only if the coefficient of $X$
is nonzero (see \cite[Theorem 3.50]{Lidl}).
We will be particulary interested in such
linearized polynomials, which merit a special name:
A monic linearized polynomial is called a {\em subspace polynomial} if it has no multiple roots
(see \cite{Ben12}, \cite{Ben10}, \cite{Cheng} or \cite{Wa}).
We remark that
a subspace polynomial
with respect to $\mathbb{F}_{q^n}$ can be defined alternatively as the annihilator polynomial of a subspace of $\mathbb{F}_{q^n}$,
in order to make sense into using the term ``subspace".
There is  an obvious one-to-one correspondence between the $k$-dimensional subspaces of $\mathcal{P}_q(N)$
and the subspace polynomials with $q$-degree $k$ whose splitting fields are $\mathbb{F}_{q^N}$.
In particular, two subspaces are identical  if and only if their
corresponding subspace polynomials are identical.
This suggests that the resolution of vector space problems can be
converted into the resolution of polynomial problems.

Given a $k$-dimensional subspace $V\in \mathcal{P}_q(N)$ and a nonzero element $\alpha\in \mathbb{F}_{q^N}$,
the subspace polynomial corresponding to  the subspace $\alpha V=\{\alpha v\,|\,v\in V\}$
has been characterized in \cite[Lemma 5]{Ben}:
\begin{lem}\label{lem1}
Let $V$ be a $k$-dimensional subspace of $\mathbb{F}_{q^N}$  and let $\alpha$ be a nonzero element of
$\mathbb{F}_{q^N}$. If
$$T(X)=\prod_{v\in V}(X-v)=X^{q^k}+\sum_{i=0}^{k-1}a_{i}X^{q^i}$$
is the subspace polynomial corresponding to $V$,
then the subspace polynomial corresponding  to $\alpha V$ is given by
$$T_\alpha(X)=\prod_{v\in V}(X-\alpha v)=X^{q^k}+\sum_{i=0}^{k-1}\alpha^{q^k-q^i}a_{i}X^{q^i}.$$
\end{lem}

\section{Constructions of  cyclic constant dimension codes}\label{sec-main}
We first propose a new approach to generalize \cite[Theorem 3]{Ben}, which can be seen as the $\ell=1$
case of the following lemma.

\begin{lem}\label{generalized1}
Let $k$ and $\ell$ be positive integers with $1\leq\ell<k$ and $\gcd(\ell,k)=1$.  Let $a_0$ and $a_\ell$ be nonzero elements of $\mathbb{F}_{q^n}$,
where $n$ is a positive integer.
Suppose that the set $V$ of roots of the subspace polynomial
$$T(X)=X^{q^k}+a_\ell X^{q^\ell}+a_0X\in \mathbb{F}_{q^n}[X]$$
is contained in $\mathbb{F}_{q^N}$.
Then
$$
\mathcal{C}=\Big\{\alpha V\,\Big{|}\, \alpha\in \mathbb{F}_{q^N}^*\Big\}
$$
is a $k$-dimensional cyclic subspace code with size $\frac{q^N-1}{q-1}$ and  minimum distance $2k-2$.
\end{lem}
\begin{proof}
It is a known fact that if  $\mathcal{C}$ has size $(q^N-1)/(q-1)$ then $\mathcal{C}$ cannot have minimum distance $2k$;
this is simply because if it were $2k$ then $\mathbb{F}_{q^N}^*=\mathbb{F}_{q^N}\setminus\{0\}$ would contain
$\frac{(q^N-1)(q^k-1)}{(q-1)}$ elements, which is impossible.
Therefore,
to obtain the desired result, it suffices to prove that
\begin{equation}\label{goal}
\dim\big(V\bigcap\alpha V\big)\leq1~~\hbox{for any $\alpha\in \mathbb{F}_{q^N}^*\setminus\mathbb{F}_{q}^*$}.
\end{equation}
Fix an element  $\alpha\in \mathbb{F}_{q^N}^*\setminus\mathbb{F}_{q}^*$.
By Lemma \ref{lem1}, the subspace polynomial corresponding to $\alpha V$
is
$$
T_\alpha(X)=X^{q^k}+a_\ell\alpha^{q^k-q^\ell} X^{q^\ell}+a_0\alpha^{q^k-1}X.
$$
Suppose now  $a$ and $b$ are any nonzero elements of $V\bigcap\alpha V$.
We aim  to show that an element  $\lambda\in \mathbb{F}_{q}^*$ can be found such that $a=\lambda b$;
we then conclude that (\ref{goal}) is achieved,  and thus the proof is complete.
To this end, we first claim that if exactly one of  $\alpha^{q^k-q^\ell}-1$ and $\alpha^{q^k-1}-1$ is equal to $0$, then
we arrive at (\ref{goal}) at once.
Indeed, observe that
$$
T_\alpha(X)-T(X)=a_\ell\Big(\alpha^{q^k-q^\ell}-1\Big) X^{q^\ell}+a_0\Big(\alpha^{q^{k}-1}-1\Big)X.
$$
If $\alpha^{q^k-q^\ell}-1\neq0$ and $\alpha^{q^{k}-1}-1=0$ (or, $\alpha^{q^k-q^\ell}-1=0$ and $\alpha^{q^{k}-1}-1\neq0$),
then the subspace polynomials $T_\alpha(X)$ and $T(X)$ have a unique common root  $0$, proving the claim.
If both
$$\alpha^{q^k-q^\ell}-1=0 ~~\hbox{and}~~ \alpha^{q^k-1}-1=0,$$
then $\alpha=\alpha^{q^k}=\alpha^{q^\ell}$, which implies that
$\alpha\in \mathbb{F}_{q^k}$ and $\alpha\in \mathbb{F}_{q^\ell}$. However, our assumption   $\gcd(\ell,k)=1$ forces
$\alpha\in \mathbb{F}_q^*$. This is a contradiction.
Thus, it cannot  occur simultaneously that $\alpha^{q^k-q^\ell}-1=0$ and $\alpha^{q^k-1}-1=0$.
We can assume, therefore, that
$\alpha^{q^k-q^\ell}-1\neq0$ and  $\alpha^{q^k-1}-1\neq0$.
Since $a$ and $b$ are contained in $V\bigcap\alpha V$, we have
$$
T(a)=T(b)=T_\alpha(a)=T_\alpha(b)=0.
$$
This leads to
$$
T_\alpha(a)-T(a)=a_\ell\Big(\alpha^{q^k-q^\ell}-1\Big) a^{q^\ell}+a_0\Big(\alpha^{q^k-1}-1\Big)a=0
$$
and
$$
T_\alpha(b)-T(b)=a_\ell\Big(\alpha^{q^k-q^\ell}-1\Big) b^{q^\ell}+a_0\Big(\alpha^{q^k-1}-1\Big)b=0.
$$
It follows that
$$
a^{q^\ell-1}=\frac{-a_0(\alpha^{q^k-1}-1)}{a_\ell(\alpha^{q^k-q^\ell}-1)}=b^{q^\ell-1},
$$
or, equivalently,
$$
\frac{a}{b}=\Big(\frac{a}{b}\Big)^{q^\ell}
$$
which gives $a/b\in \mathbb{F}_{q^\ell}$.  Let $a/b=\lambda\in \mathbb{F}_{q^\ell}$, namely $a=b\lambda$.
By $T(a)=0$ and $a=\lambda b$, we have
$$
0=T(a)=a^{q^k}+a_\ell a^{q^\ell}+a_0a=\lambda^{q^k}b^{q^k}+a_\ell\lambda^{q^\ell}b^{q^\ell}+a_0b\lambda
=\lambda^{q^k}b^{q^k}+a_\ell\lambda b^{q^\ell}+a_0b\lambda,
$$
where the last equality holds because $\lambda$ is an element of $\mathbb{F}_{q^\ell}$.
Combining with
$$\lambda T(b)=\lambda b^{q^k}+\lambda a_\ell b^{q^\ell}+\lambda a_0b=0,
$$
we have
$\lambda=\lambda^{q^k}$.  By $\gcd(\ell,k)=1$ again, we finally conclude that  $\lambda\in \mathbb{F}_q$, as wanted.
The proof is complete.
\end{proof}

As pointed out by \cite[Section 2.2]{OO16}, the  value  of  $N$ in Lemma \ref{generalized1} cannot be chosen freely; it is depending
on the values of $k,\ell,n$ and the nonzero elements $a_\ell,a_0$.
We include the following example to illustrate Lemma \ref{generalized1}.
\begin{Example}\rm
We adopt the notation in Lemma \ref{generalized1}.
Take $q=3$, $n=1$,   and  $k=5$. Consider the degree  $N'$ of the splitting field   of the polynomial
$X^{3^5}+a_{\ell}X^{3^\ell}+a_0X\in \mathbb{F}_3[X]$, where $1\leq \ell\leq 4$ and $a_{\ell},a_0\in\{1,-1\}$; the values of $N'$
can be determined easily by using the computer algebra system GAP \cite{GAP}, as exhibited in Table $3.1$.
Lemma \ref{goal} ensures that
there exists a $5$-dimensional cyclic subspace code of size $\frac{3^N-1}{2}$ and  minimum distance $8$ in $\mathbb{F}_{3^N}$
when  $N$ is a multiple of $N'$. For instance,  the first row of Table $3.1$ implies that the set of roots of the subspace polynomial
$X^{3^5}+X^{3}+X$ forms a $5$-dimensional cyclic subspace code of size $\frac{3^N-1}{2}$ and  minimum distance $8$  in $\mathbb{F}_{3^N}$
when  $N$ is a multiple of $78$. We then compare among the last column of Table $3.1$ to pick out the minimal elements with respect to
the partially ordered by divisibility (for positive integers $a,b$, $a\leq b$ precisely when $a$ divides $b$).
After a bit of simple calculations, the minimal elements are
$78,121,$ $80$ and $104$.
From the first four rows of Table $3.1$, one sees that \cite[Theorem 3]{Ben} and \cite[Theorem 3]{OO16} only produce
the first two values $78$ and $121$.
This example suggests that Lemma \ref{goal} indeed could provide subspace codes for more various values of $N$.
\begin{table}[!h]\begin{center}\label{table1232}
\caption{The degrees of the  splitting fields of the polynomials $X^{3^5}+a_{\ell}X^{3^\ell}+a_0X$}
\begin{tabular}{ c|c|c|c  }\hline
Values of $\ell$ & $(a_\ell,a_0)$ & Polynomials     &  The degrees $N'$ of the splitting fields over $\mathbb{F}_3$     \\\hline
$1$& $(1,1)$&$X^{3^5}+X^{3}+X$   &  $78$   \\
$1$& $(1,-1)$&$X^{3^5}+X^{3}-X$    &  $78$ \\
$1$& $(-1,1)$&$X^{3^5}-X^{3}+X$    &  $242$ \\
$1$& $(-1,-1)$& $X^{3^5}-X^{3}-X$    &  $121$ \\
$2$& $(1,1)$&$X^{3^5}+X^{3^2}+X$   &  $80$ \\
$2$& $(1,-1)$&$X^{3^5}+X^{3^2}-X$    &  $104$ \\
$2$& $(-1,1)$&$X^{3^5}-X^{3^2}+X$    &  $312$ \\
$2$& $(-1,-1)$& $X^{3^5}-X^{3^2}-X$    &  $80$ \\
$3$& $(1,1)$&$X^{3^5}+X^{3^3}+X$   &  $80$ \\
$3$& $(1,-1)$&$X^{3^5}+X^{3^3}-X$    &  $80$ \\
$3$& $(-1,1)$&$X^{3^5}-X^{3^3}+X$    &  $312$ \\
$3$& $(-1,-1)$& $X^{3^5}-X^{3^3}-X$    &  $104$ \\
$4$& $(1,1)$&$X^{3^5}+X^{3^4}+X$   &  $78$ \\
$4$& $(1,-1)$&$X^{3^5}+X^{3^4}-X$    &  $121$ \\
$4$& $(-1,1)$&$X^{3^5}-X^{3^4}+X$    &  $242$ \\
$4$& $(-1,-1)$& $X^{3^5}-X^{3^4}-X$    &  $78$ \\
\hline
\end{tabular}
\end{center}\end{table}
\end{Example}

Some  cyclic constant dimension codes produced by Lemma \ref{generalized1} can be put together to form a larger code,
but without reducing the minimum distance. The following lemma is a generalization of \cite[Theorem 3]{OO16}, which considers the case $\ell=1$.
\begin{lem}\label{theorem2}
Let $k$ and $\ell$ be positive integers with $1\leq\ell<k$ and $\gcd(\ell,k)=1$, and let
$$
T^{(i)}(X)=X^{q^k}+\theta_i X^{q^\ell}+\gamma_iX\in \mathbb{F}_{q^n}[X]
$$
be $r$ subspace polynomials over $\mathbb{F}_{q^n}$ with $\theta_i$ and $\gamma_i$ being
nonzero elements of $\mathbb{F}_{q^n}$ for $1\leq i\leq r$.
Suppose that $V_i$ is the set of roots of the subspace polynomial $T^{(i)}(X)$ and that
$V_i$ {\rm(}$1\leq i\leq r${\rm)} are contained in $\mathbb{F}_{q^N}$.
If
\begin{equation}\label{assumption}
\Big(\frac{\gamma_i}{\gamma_j}\Big)^{\frac{q^{\ell}-1}{q-1}}\neq\Big(\frac{\gamma_i}{\gamma_j}\big(\frac{\theta_i}{\theta_j}\big)^{-1}\Big)^{\frac{q^{k}-1}{q-1}}
~~\hbox{for any $1\leq i\neq j\leq r$},
\end{equation}
then
$$
\mathcal{C}=\bigcup_{i=1}^r\Big\{\alpha V_i\,\Big{|}\, \alpha\in \mathbb{F}_{q^N}^*\Big\}
$$
is a $k$-dimensional cyclic subspace code of size $r\frac{q^N-1}{q-1}$ and   minimum distance $2k-2$.
\end{lem}
\begin{proof}
Using   Lemma \ref{generalized1}, it is enough to prove that
\begin{equation}\label{goa2}
\dim\big(V_i\bigcap\alpha V_j\big)\leq1~~\hbox{for any $\alpha\in \mathbb{F}_{q^N}^*$  and  $1\leq i\neq j\leq r$}.
\end{equation}
The proof is similar to that of Lemma  \ref{generalized1},
with a few modifications.
Let $\alpha\in \mathbb{F}_{q^N}^*$  and let $1\leq i\neq j\leq r$ be two distinct integers.
The subspace polynomial corresponding to $\alpha V_j$
is
$$
T_\alpha^{(j)}(X)=X^{q^k}+\theta_j\alpha^{q^k-q^\ell} X^{q^\ell}+\gamma_j\alpha^{q^k-1}X,
$$
and  thus
$$
T_\alpha^{(j)}(X)-T^{(i)}(X)=\Big(\theta_j\alpha^{q^k-q^\ell}-\theta_i\Big) X^{q^\ell}+\Big(\gamma_j\alpha^{q^k-1}-\gamma_i\Big)X.
$$
As we did  in the proof of Lemma \ref{goal},
we are done if exactly one of  $\theta_j\alpha^{q^k-q^\ell}-\theta_i$ and $\gamma_j\alpha^{q^k-1}-\gamma_i$ is equal to $0$.
If
$$
\theta_j\alpha^{q^k-q^\ell}-\theta_i=0
~~\hbox{
and}
~~
\gamma_j\alpha^{q^k-1}-\gamma_i=0
$$
then
$$
\alpha^{q^\ell-1}=\frac{\gamma_i}{\gamma_j}\Big(\frac{\theta_i}{\theta_j}\Big)^{-1}
~~\hbox{
and}
~~
\alpha^{q^k-1}=\frac{\gamma_i}{\gamma_j}.
$$
This leads to
$$
\alpha^{\frac{(q^\ell-1)(q^k-1)}{q-1}}=\Big(\frac{\gamma_i}{\gamma_j}\big(\frac{\theta_i}{\theta_j}\big)^{-1}\Big)^{\frac{q^k-1}{q-1}}
$$
and
$$
\alpha^{\frac{(q^\ell-1)(q^k-1)}{q-1}}=\Big(\frac{\gamma_i}{\gamma_j}\Big)^{\frac{q^\ell-1}{q-1}}.
$$
Therefore, one has
$$
\Big(\frac{\gamma_i}{\gamma_j}\big(\frac{\theta_i}{\theta_j}\big)^{-1}\Big)^{\frac{q^k-1}{q-1}}=\alpha^{\frac{(q^\ell-1)(q^k-1)}{q-1}}=
\Big(\frac{\gamma_i}{\gamma_j}\Big)^{\frac{q^\ell-1}{q-1}},
$$
which  contradicts our
assumption (\ref{assumption}).
We can assume, therefore, that
$$
\theta_j\alpha^{q^k-q^\ell}-\theta_i\neq0~~\hbox{and}~~\gamma_j\alpha^{q^k-1}-\gamma_i\neq0
.
$$
At this point, taking arguments similar to those used in the proof of Lemma \ref{generalized1}, we obtained the desired result.
\end{proof}

The following corollaries are   direct   consequence of Lemma \ref{theorem2}.
We first specialize Lemma \ref{theorem2} to the case of $\ell=1$ and $\theta_i=\gamma_i$.
\begin{Corollary}\label{uselater1}
Let $r$ be an integer with  $1\leq r\leq q^n-1$ and let
$$
T^{(i)}(X)=X^{q^k}+\theta_i X^{q}+\theta_iX\in \mathbb{F}_{q^n}[X]
$$
be $r$ subspace polynomials over $\mathbb{F}_{q^n}$, where  $\theta_i$  are distinct nonzero elements
of  $\mathbb{F}_{q^n}$ for $1\leq i\leq r$.
Suppose that $V_i$ is the set of roots of the
subspace polynomial $T^{(i)}(X)$ and that
$V_i$ {\rm(}$1\leq i\leq r${\rm)} are contained in $\mathbb{F}_{q^N}$.
Then
$$
\mathcal{C}=\bigcup_{i=1}^r\Big\{\alpha V_i\,\Big{|}\, \alpha\in \mathbb{F}_{q^N}^*\Big\}
$$
is a $k$-dimensional cyclic subspace code of size $r\frac{(q^N-1)}{q-1}$ and  minimum distance $2k-2$.
\end{Corollary}
\begin{proof}
Take $\ell=1$    and $\theta_i=\gamma_i$ for each $1\leq i\leq r$ in Lemma \ref{theorem2}. The right-hand side of
inequality (\ref{assumption})
is equal to $1$; however, the left-hand side of (\ref{assumption}) is certainly not equal to $1$.
It follows that inequality (\ref{assumption}) holds true, and we get the desired result by applying Lemma \ref{theorem2}.
\end{proof}

The following corollary aims to provide an example of large codes for the $\ell\neq1$ case;
however, in order to state the conditions compactly,   we restrict ourself to the case where $q=2$ and $n=k$.
\begin{Corollary}\label{uselater2}
Let $k$,  $\ell$, and $r$ be positive integers with $1\leq\ell<k$,  $\gcd(\ell,k)=1$ and $1\leq r\leq 2^k-1$. Let
$$
T^{(i)}(X)=X^{2^k}+\theta_i X^{2^\ell}+\theta_iX\in \mathbb{F}_{2^k}[X]
$$
be $r$ subspace polynomials over $\mathbb{F}_{2^k}$ with $\theta_i$  being
distinct nonzero elements of $\mathbb{F}_{2^k}$ for $1\leq i\leq r$.
Suppose that $V_i$ is the set of roots of
the subspace polynomial $T^{(i)}(X)$ and that
$V_i$ {\rm(}$1\leq i\leq r${\rm)} are contained in $\mathbb{F}_{2^N}$.
Then
$$
\mathcal{C}=\bigcup_{i=1}^r\Big\{\alpha V_i\,\Big{|}\, \alpha\in \mathbb{F}_{2^N}^*\Big\}
$$
is a  $k$-dimensional cyclic subspace code of size $r(2^N-1)$ and  minimum distance $2k-2$.
\end{Corollary}
\begin{proof}
Taking $q=2$, $n=k$ and $\theta_i=\gamma_i$ for each $1\leq i\leq r$ in  Lemma \ref{theorem2}, we are left to show that
\begin{equation*}
\Big(\frac{\theta_i}{\theta_j}\Big)^{2^\ell-1}\neq1
~~\hbox{for any $1\leq i\neq j\leq r$}.
\end{equation*}
Since $\theta_i$ are distinct nonzero elements of $\mathbb{F}_{2^k}$ for $1\leq i\leq r$, we have
$$
\frac{\theta_i}{\theta_j}\neq1~~\hbox{and}~~\Big(\frac{\theta_i}{\theta_j}\Big)^{2^k-1}=1
$$
for any $1\leq i\neq j\leq r$. If
$(\theta_i/\theta_j)^{2^\ell-1}$ were equal to $1$, we would have $\theta_i/\theta_j=1$ since $\gcd(2^k-1, 2^\ell-1)=1$.
This is a contradiction.
We are done.
\end{proof}

Here is an example to illustrate Corollary \ref{uselater2}.
\begin{Example}\label{example-cor5}\rm
Take $k=5$ and $r=5$ in Corollary \ref{uselater2}.
Let $\theta$ be a generator of the cyclic group $\mathbb{F}_{2^5}^*$ given by the computer algebra system GAP \cite{GAP}.
Set $\theta_1=\theta^3$, $\theta_2=\theta^6$, $\theta_3=\theta^{12}$, $\theta_4=\theta^{17}$, and
$\theta_5=\theta^{24}$.
Fix  a value $\ell$,  $1\leq \ell\leq 4$.
Let $N_\ell'$ denote the degree of the splitting field of the polynomials
$
T^{(i)}(X)=X^{2^5}+\theta_i X^{2^\ell}+\theta_iX\in \mathbb{F}_{2^5}[X],
$
$1\leq i\leq 5$.
With $\ell$ ranging from $1$ to $4$,
Corollary \ref{uselater2} then permits us to produce $4$ constant  cyclic subspace codes in $\mathbb{F}_{2^{N}}$  having size
$5(2^{N}-1)$ and minimum distance  $2k-2=8$, where $N$ is a multiple of $N_\ell'$.
Using GAP \cite{GAP}, the values of $N_\ell'$ are listed in Table $3.2$.
\begin{table}[!h]\begin{center}\label{table12322}
\caption{The values of $N_\ell'$}
\begin{tabular}{ c|c }\hline
Values of $\ell$     &  Values of $N_\ell'$     \\\hline
$1$    &  $30$   \\
$2$    &  $70$ \\
$3$     &  $75$ \\
$4$      &  $60$ \\
\hline
\end{tabular}
\end{center}\end{table}
\end{Example}

Gluesing-Luerssen {\it et al.}  studied constant cyclic subspace codes having full minimum distances in \cite{Glu};
it is well known that the set of all cyclic
shifts of $\mathbb{F}_{q^k}$ (as a subfield of $\mathbb{F}_{q^N}$) forms a cyclic subspace code  of full minimum distance $2k$.
To insert such codes into those produced in Lemma \ref{theorem2} (where codes are characterized through subspace polynomials),
the codes having full minimum distance are also described by  subspace polynomials for consistency, as we show below.
\begin{Proposition}\label{bi}
Let $k>1$   be a positive integer and let $a_0$   be  a nonzero element  of $\mathbb{F}_{q^n}$.
Suppose that the set $U$ of roots of  the subspace polynomial
$$T(X)=X^{q^k}-a_0X\in \mathbb{F}_{q^n}[X]$$
is contained in $\mathbb{F}_{q^N}$.
Then
$$
\mathcal{C}=\Big\{\alpha U\,\Big{|}\, \alpha\in \mathbb{F}_{q^N}^*\Big\}
$$
is a $k$-dimensional cyclic subspace code of size $\frac{q^N-1}{q^k-1}$ and  minimum distance $2k$.
\end{Proposition}
\begin{proof}
Let $\alpha\in \mathbb{F}_{q^N}^*$. The subspace polynomial corresponding to $\alpha U$
is
$$
T_\alpha(X)=X^{q^k}-a_0\alpha^{q^k-1}X.
$$
It is readily seen that
$$
T_\alpha(X)-T(X)=a_0\Big(1-\alpha^{q^k-1}\Big)X.
$$
Hence,  $\alpha U=U$ if and only if $\alpha^{q^k-1}=1$,
which implies that the size of $\mathcal{C}$ is equal to $(q^N-1)/(q^k-1)$.
Finally, it is trivial to see that the minimum distance of $\mathcal{C}$ is $2k$.
We are done.
\end{proof}
\begin{Remark}\rm
Theoretically, the splitting field of any binomial over a finite field can be determined easily. Indeed,
suppose $a_0$ has order $s$ in the  cyclic group $\mathbb{F}_{q^n}^*$.
Then the degree of the splitting field $N'$ of
$T(X)=X^{q^k}-a_0X\in \mathbb{F}_{q^n}[X]$
is equal to the multiplicative order of $q$ modulo $s(q^k-1)$, i.e., $N'$ is the smallest positive integer such that
$s(q^k-1)$ divides $q^{N'}-1$. In fact, a primitive $s(q^k-1)$-th root of unity $\omega$ in the finite field $\mathbb{F}_{q^{N'}}$
can be found such that $\omega^{q^k-1}=a_0$. It is clear that $k$ is a divisor of $N'$.
\end{Remark}

We present the following example to illustrate Proposition \ref{bi}.

\begin{Example}\rm
Let $k=5$, $q=3$,  and $n=5$ in Proposition \ref{bi}.
Let $a_0$ be an  element of $\mathbb{F}_{3^5}^*$ having  order $s=11$.
GAP \cite{GAP} computations show that the multiplicative order of  $q=3$ modulo $s(q^k-1)=11(3^5-1)=11\times242$ is equal to
$N'=55$. It follows that
the splitting field of the subspace polynomial
$$
T(X)=X^{3^5}-a_0X\in \mathbb{F}_{3^5}[X]
$$
is $\mathbb{F}_{3^{55}}$.  Let $U$ be the set of  roots of  $T(X)$, and it follows from Proposition \ref{bi}  that
$$
\mathcal{C}=\Big\{\alpha U\,\Big{|}\, \alpha\in \mathbb{F}_{3^{55}}^*\Big\}
$$
is a cyclic subspace  code of size $\frac{q^{55}-1}{q^5-1}=\frac{3^{55}-1}{3^5-1}$ and minimum distance $2k=10$.
\end{Example}

At the moment,
we derive the following theorem which puts certain cyclic constant dimension codes generated by Lemma
\ref{theorem2} and Proposition \ref{bi} together to form a new code.
\begin{Theorem}\label{theoremend}
Let $k$ and $\ell$ be positive integers with $1\leq\ell<k$ and $\gcd(\ell,k)=1$. Let
$$
T(X)=X^{q^k}-a_0 X\in \mathbb{F}_{q^n}[X]
$$
and
$$
T^{(i)}(X)=X^{q^k}+\theta_i X^{q^\ell}+\gamma_iX\in \mathbb{F}_{q^n}[X]
$$
be $r+1$ subspace polynomials over $\mathbb{F}_{q^n}$ with $a_0$, $\theta_i$,  and $\gamma_i$ being
nonzero elements of $\mathbb{F}_{q^n}$ for $1\leq i\leq r$.
Suppose that $U$ and  $V_i$ are the sets of roots of
the subspace polynomials $T(X)$ and $T^{(i)}(X)$, respectively,
and that   $U$ and $V_i$ are contained in $\mathbb{F}_{q^N}$ for all $1\leq i\leq r$.
If
\begin{equation*}
\Big(\frac{\gamma_i}{\gamma_j}\Big)^{\frac{q^{\ell}-1}{q-1}}\neq\Big(\frac{\gamma_i}{\gamma_j}\big(\frac{\theta_i}{\theta_j}\big)^{-1}\Big)^{\frac{q^{k}-1}{q-1}}
~~\hbox{for any $1\leq i\neq j\leq r$}
\end{equation*}
then
$$
\mathcal{C}=\Big(\bigcup_{i=1}^r\Big\{\alpha V_i\,\Big{|}\, \alpha\in \mathbb{F}_{q^N}^*\Big\}\Big)
\bigcup\Big\{\alpha U\,\Big{|}\, \alpha\in \mathbb{F}_{q^N}^*\Big\}
$$
is a $k$-dimensional cyclic subspace code of size $r\frac{q^N-1}{q-1}+\frac{q^N-1}{q^k-1}$ and  minimum distance $2k-2$.
\end{Theorem}
\begin{proof}
By Lemma \ref{theorem2} and Proposition \ref{bi}, it is enough to show that
\begin{equation*}
\dim\big(\alpha U\bigcap V_i\big)\leq1~~\hbox{for any $\alpha\in \mathbb{F}_{q^N}^*$  and  $1\leq i\leq r$}.
\end{equation*}
Fix a nonzero element  $\alpha\in \mathbb{F}_{q^N}$. The subspace polynomial corresponding to $\alpha U$
is
$$
T_\alpha(X)=X^{q^k}-a_0\alpha^{q^k-1}X.
$$
Clearly,
$$
T^{(i)}(X)-T_\alpha(X)=\theta_i X^{q^\ell}+\big(\gamma_i+a_0\alpha^{q^k-1}\big)X\in \mathbb{F}_{q^n}[X].
$$
If $\gamma_i+a_0\alpha^{q^k-1}=0$ then the desired result holds trivially. Suppose  $\gamma_i+a_0\alpha^{q^k-1}\neq0$.
Let    $a$ and $b$ be any nonzero elements of $\alpha U\bigcap V_i$.
We then have
$$
T^{(i)}(a)=T^{(i)}(b)=T_\alpha(a)=T_\alpha(b)=0.
$$
Hence,
$$
T^{(i)}(a)-T_\alpha(a)=\theta_i  a^{q^\ell}+\Big(\gamma_i+a_0\alpha^{q^k-1}\Big)a=0
$$
and
$$
T^{(i)}(b)-T_\alpha(b)=\theta_i  b^{q^\ell}+\Big(\gamma_i+a_0\alpha^{q^k-1}\Big)b=0.
$$
It follows that
$$
a^{q^\ell-1}=\frac{\gamma_i+a_0\alpha^{q^k-1}}{-\theta_i}=b^{q^\ell-1}.
$$
Therefore,
$$
\frac{a}{b}=\Big(\frac{a}{b}\Big)^{q^\ell},
$$
which gives $a/b\in \mathbb{F}_{q^\ell}$.  Let $a/b=\lambda\in \mathbb{F}_{q^\ell}$, or alternatively,  $a=b\lambda$.
As we have done in the proof of Lemma \ref{generalized1}, we conclude that $\lambda$ is, in fact,  contained in $\mathbb{F}_{q}^*$.
The proof is complete.
\end{proof}

We end this section with the following two immediate corollaries of Theorem \ref{theoremend} and Corollaries \ref{uselater1} and \ref{uselater2}.
\begin{Corollary}\label{qn}
Let $k>1$ be a positive integer.
Let $n$ and $r$ be positive integers with  $1\leq r\leq q^n-1$.
Let
$$
T(X)=X^{q^k}-a_0 X\in \mathbb{F}_{q^n}[X]
$$
and
$$
T^{(i)}(X)=X^{q^k}+\theta_i X^{q}+\theta_iX\in \mathbb{F}_{q^n}[X]
$$
be $r+1$  subspace polynomials over $\mathbb{F}_{q^n}$ with $a_0$ and $\theta_i$    being
nonzero elements of $\mathbb{F}_{q^n}$ for $1\leq i\leq r$.
Suppose that $U$ and  $V_i$ are the sets of roots of  $T(X)$ and $T^{(i)}(X)$, respectively,
and that  $U$ and $V_i$, $1\leq i\leq r$, are contained in $\mathbb{F}_{q^N}$.
Then
$$
\mathcal{C}=\Big(\bigcup_{i=1}^r\Big\{\alpha V_i\,\Big{|}\, \alpha\in \mathbb{F}_{q^N}^*\Big\}\Big)
\bigcup\Big\{\alpha U\,\Big{|}\, \alpha\in \mathbb{F}_{q^N}^*\Big\}
$$
is a  $k$-dimensional cyclic subspace code of size $r\frac{q^N-1}{q-1}+\frac{q^N-1}{q^k-1}$ and  minimum distance $2k-2$.
\end{Corollary}

Here is an example to demonstrate Corollary \ref{qn}.
\begin{Example}\rm
Take $q=3$, $k=5$, $n=1$, and $r=2$ in Corollary \ref{qn}. Let $\theta_1=1$ and $\theta_2=-1$.
GAP \cite{GAP} computations show that the degrees of the splitting fields of $T^{(1)}(X)$ and $T^{(2)}(X)$ are equal to
$78$ and $121$, respectively. We simply take $a_0=1$, thus the degree of the splitting field of $T(X)$ is equal to
$5$.
Setting $N=78\times121\times5$, one knows that the vector spaces $U$, $V_1$ and $V_2$ are contained in $\mathbb{F}_{3^N}$.
It follows from Corollary \ref{qn} that
$$
\mathcal{C}=\Big(\bigcup_{i=1}^2\Big\{\alpha V_i\,\Big{|}\, \alpha\in \mathbb{F}_{3^N}^*\Big\}\Big)
\bigcup\Big\{\alpha U\,\Big{|}\, \alpha\in \mathbb{F}_{3^N}^*\Big\}
$$
is a $5$-dimensional cyclic subspace code of size $3^N-1+\frac{3^N-1}{3^5-1}$ and  minimum distance $8$.
\end{Example}

\begin{Corollary}\label{last-cor}
Let $k$,  $\ell$, and $r$ be positive integers with $1\leq\ell<k$,  $\gcd(\ell,k)=1$ and $1\leq r\leq 2^k-1$. Let
$$
T(X)=X^{2^k}+a_0 X\in \mathbb{F}_{2^k}[X]
$$
and
$$
T^{(i)}(X)=X^{2^k}+\theta_i X^{2^\ell}+\theta_iX\in \mathbb{F}_{2^k}[X]
$$
be $r+1$ subspace polynomials over $\mathbb{F}_{2^k}$ with $a_0$ and $\theta_i$     being
nonzero elements of $\mathbb{F}_{2^k}$ for $1\leq i\leq r$.
Suppose that $U$ and  $V_i$ are the sets of roots of  $T(X)$ and $T^{(i)}(X)$, respectively,
and that  $U$ and $V_i$, $1\leq i\leq r$, are contained in $\mathbb{F}_{2^N}$.
Then
$$
\mathcal{C}=\Big(\bigcup_{i=1}^r\Big\{\alpha V_i\,\Big{|}\, \alpha\in \mathbb{F}_{2^N}^*\Big\}\Big)
\bigcup\Big\{\alpha U\,\Big{|}\, \alpha\in \mathbb{F}_{2^N}^*\Big\}
$$
is a $k$-dimensional cyclic subspace code of size $r(2^N-1)+\frac{2^N-1}{2^k-1}$ and  minimum distance $2k-2$.
\end{Corollary}
We present an illustrative example of  Corollary \ref{last-cor}.
\begin{Example}\rm
The codes generated by Example \ref{example-cor5} are enlarged here to produce bigger codes  without compromising the minimum distances.
Take $k=5$ and $r=5$ in Corollary \ref{last-cor}. Recall from Example \ref{example-cor5} that
we have already obtained  four $5$-dimensional cyclic subspace codes in $\mathbb{F}_{2^{N}}$, each of which has size $5(2^{N}-1)$
and minimum distance $8$. Take $a_0=1$, then $\mathbb{F}_{2^5}$ is the splitting field for $X^{2^5}+X.$  Corollary
\ref{last-cor}  thus gives us four $5$-dimensional cyclic subspace code in $\mathbb{F}_{2^{N}}$, each of which has size
$5(2^{N}-1)+(2^{N}-1)/31$ and minimum distance $8$.

\end{Example}

\section{Concluding remarks}\label{sec-concluding}
In this paper,  we study the construction of $k$-dimensional cyclic subspace codes with minimum distance $2k-2$
by  exploring further the ideas   proposed in
\cite{Ben} and \cite{OO16}.
Lemma \ref{generalized1} is the key ingredient in computing the minimum distance of cyclic subspace codes described by
a special class of subspace polynomials.
Our main result, Theorem \ref{theoremend}, improves \cite[Theorem 3]{OO16} in two directions:
First we introduce a parameter $\ell$ which is a positive integer small than and coprime to $k$,
thus \cite[Theorem 3]{Ben} and \cite[Theorem 3]{OO16} can bee seen as the  special case $\ell=1$ of our result; second we enlarge
the code size
by adjoining  with a spread code, without compromising the minimum distance.
However, the conjecture raised in \cite{Tra} and \cite{Glu} (see Section 1)  is still open.
It is interesting to find new tools or combine several tools
in the literature in order to solve this problem.
It would also be a happy outcome of this paper
if one can generalize the  method to get further results.

\vskip 5mm
\noindent{\bf Acknowledgments}\quad\small
We sincerely thank the Associate Editor and the anonymous referees
for their carefully reading and helpful suggestions
which led to significant improvements of the paper.
The research of Bocong Chen is supported by NSFC (Grant No.
11601158) and the Fundamental Research Funds for the Central Universities (Grant No. 2017MS111).
The research of Hongwei Liu is supported by NSFC (Grant No.
11171370) and self-determined research funds of CCNU from the colleges' basic research and operation of
MOE (Grant No. CCNU14F01004).

\end{document}